\documentclass{amsart}
\usepackage{amssymb, amscd, amsfonts}
\usepackage[mathscr]{eucal}
\usepackage{mathrsfs}
\usepackage{graphicx}
\usepackage{subfigure}
\usepackage{lscape}

\newtheorem{theorem}{Theorem}[section]
\newtheorem{lemma}[theorem]{Lemma}

\newtheorem{proposition}[theorem]{Proposition}

\theoremstyle{definition}
\newtheorem{definition}[theorem]{Definition}

\newtheorem{claim}{Claim}

\theoremstyle{remark}
\newtheorem*{remark}{Remark}

\newcommand{\remove}[1]{}

\newcommand{\pa}[2]{\mathbf{P}_{#1}^{(#2)}}
\newcommand{\cy}[2]{\mathbf{Q}_{#1}^{(#2)}}
\newcommand{\hp}[2]{\mathbf{H}_{#1}^{(#2)}}
\newcommand{\hc}[2]{\mathbf{M}_{#1}^{(#2)}}
\newcommand{\G}{\mathbf{G}}

\newcommand{\eg}{\textit{e.g.}}
\newcommand{\ie}{\textit{i.e.}}

\begin{document}

\title{On the independent subsets of powers of paths and cycles}


\author{Pietro Codara, Ottavio M. D'Antona}




\address[P. Codara, O. M. D'Antona]
{Dipartimento di Informatica, Universit\`{a} degli Studi di Milano,
via Comelico 39/41, I-20135, Milan, Italy}

\email{\{codara,dantona\}@dico.unimi.it}
\thanks{This research was partially supported by ``Dote ricerca'' -- FSE, Regione Lombardia.}



\begin{abstract}
In the first part of this work we provide a formula for the number of edges of
the Hasse diagram of the independent subsets of
the $h$\textsuperscript{th} power of a path ordered by
inclusion. For $h=1$ such a value is the number of edges of a Fibonacci
cube. We show that, in general, the number of edges of the diagram
is obtained by convolution of a Fibonacci-like sequence with itself.

In the second part we consider the case of cycles.
We evaluate the number of edges of the Hasse diagram of the independent
subsets of the $h$\textsuperscript{th} power of a cycle ordered by
inclusion. For $h=1$, and $n>1$, such a value is the number of edges of a Lucas
cube.
\end{abstract}

\maketitle

\section{Introduction}
For a graph $\G$ we denote by $V(\G)$ the set of its vertices, and by $E(\G)$
the set of its edges.
\begin{definition}\label{def:h-path}\label{def:h-cycle}
For $n, h\geq0$,
\begin{itemize}
\item[(i)] the \emph{$h$-power of a path}, denoted by $\pa{n}{h}$, is a graph
with $n$ vertices $v_{1}$, $v_{2}$, $\dots$, $v_{n}$ such that, for $1\leq i,j\leq n$, $i\neq j$,
$(v_i,v_j)\in E(\pa{n}{h})$ if and only if $|j-i|\leq h$;
\item[(ii)] the \emph{$h$-power of a cycle}, denoted by $\cy{n}{h}$, is a graph
with $n$ vertices $v_{1}$, $v_{2}$, $\dots$, $v_{n}$ such that, for $1\leq i,j \leq n$, $i\neq j$,
$(v_i,v_j)\in E(\cy{n}{h})$ if and only if $|j-i|\leq h$ or $|j-i|\geq n-h$.
\end{itemize}
\end{definition}
Thus, for instance, $\pa{n}{0}$ and $\cy{n}{0}$ are the graphs made of $n$ isolated nodes,
$\pa{n}{1}$ is the path with $n$ vertices, and  $\cy{n}{1}$ is the cycle with $n$ vertices.
Figures \ref{fig:P_1-6_2}, and \ref{fig:Q_1-6_2} show
some powers of paths and cycles, respectively.

\begin{figure}
 \centering
 \subfigure[The graphs $\pa{1}{2}, \dots, \pa{5}{2}$]
   {\label{fig:P_1-6_2}\includegraphics[scale=0.7]{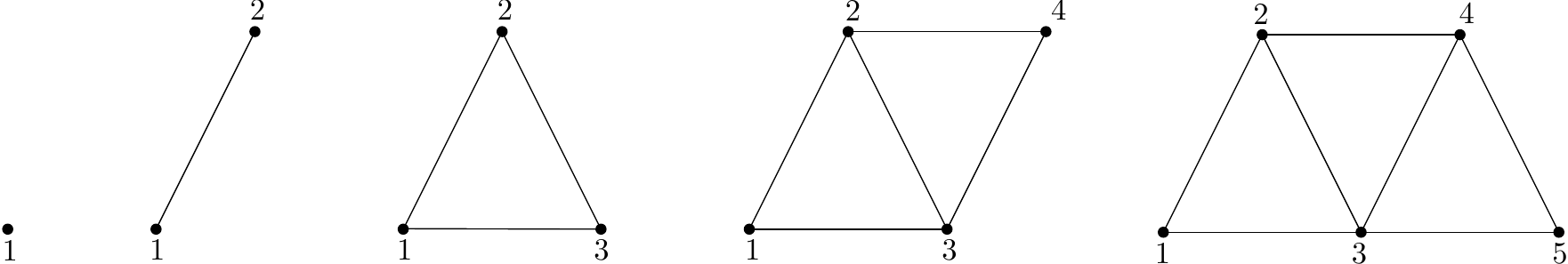}}\\
 \subfigure[The graphs $\cy{1}{2}, \dots, \cy{5}{2}$]
   {\label{fig:Q_1-6_2}\includegraphics[scale=0.75]{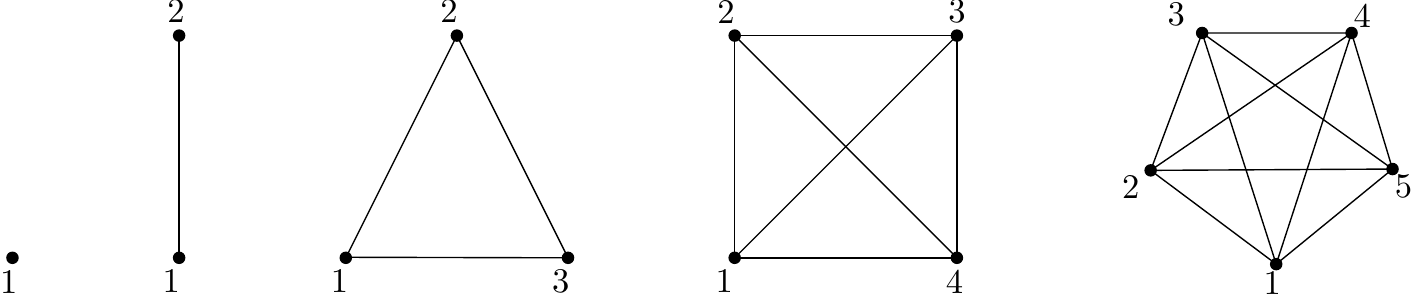}}
 \caption{Some powers of paths and cycles.}
 \end{figure}

\begin{definition}\label{def:indipendent subset}
An \emph{independent subset of a graph $\G$} is a subset of $V(\G)$
not containing adjacent vertices.
\end{definition}

Let $\hp{n}{h}$, and $\hc{n}{h}$ be the Hasse diagrams  of the posets of independent subsets of $\pa{n}{h}$, and $\cy{n}{h}$, respectively,
ordered by inclusion. Clearly, $\hp{n}{0}\cong\hc{n}{0}$ is a Boolean lattice with $n$ atoms ($n$-cube, for short).

\medskip
Every independent subset
$S$ of $\pa{n}{h}$ can be represented by a binary string $b_1 b_2 \cdots b_n$,
where, for $i=1,\dots,n$, $b_i = 1$ if and only if
$v_i \in S$.
Specifically, each
independent subset of $\pa{n}{h}$ is associated with a binary string of length
$n$ such that the distance between any two $1$'s of the string is greater than $h$.
Following \cite{fibocubes_enum} (see also \cite{survey_fibo}),
a \emph{Fibonacci string of order $n$} is a binary strings of length $n$ without
(two) consecutive $1$'s. Recalling that the Hamming distance between two binary strings $\alpha$ and $\beta$
is the number $H(\alpha,\beta)$ of bits where $\alpha$ and $\beta$ differ, we can define
the \emph{Fibonacci cube of order $n$}, denoted $\Gamma_n$, as the graph $(V,E)$, where $V$ is the set of all
Fibonacci strings of order $n$ and, for all $\alpha,\beta\in V$, $(\alpha,\beta)\in E$ if and only if $H(\alpha,\beta)=1$.
One can observe that for $h=1$ the binary strings associated with independent subsets of $\pa{n}{h}$
are \emph{Fibonacci strings of order $n$}, and the Hasse diagram of the set of all such strings ordered
bitwise is $\Gamma_n$. Fibonacci cubes were
introduced as an interconnection scheme for multicomputers in \cite{hsu}, and their
combinatorial structure has been further investigated, \eg\, in \cite{klavzar,fibocubes_enum}.
Several generalizations of the notion of Fibonacci cubes has been proposed
(see, \eg, \cite{gen_fibo,survey_fibo}).

\begin{remark}Consider the \emph{generalized Fibonacci cubes} described in \cite{gen_fibo}, \ie, the graphs $B_n(\alpha)$ obtained from the $n$-cube $B_n$ of all binary strings of length $n$ by removing all vertices that contain the binary string $\alpha$ as a substring. In this notation
the Fibonacci cube is $B_n(11)$. It is not difficult to see that $\hp{n}{h}$ cannot be expressed, in general, in terms of
$B_n(\alpha)$. Instead we have:

\centerline{
$\hp{n}{2}=B_n(11)\cap B_n(101)\,,\
\hp{n}{3}=B_n(11)\cap B_n(101)\cap B_n(1001)\,,\ \dots\,,$}
\noindent where $B_n(\alpha)\cap B_n(\beta)$ is the subgraph of $B_n$ obtaining by removing
all strings that contain either $\alpha$ or $\beta$.
\end{remark}

A similar argument can be carried out in the case of cycles. Indeed,
every independent subset
$S$ of $\cy{n}{h}$ can be represented by a circular binary string (\ie, a sequence of $0$'s and $1$'s with the first and last bits considered to be adjacent) $b_1 b_2 \cdots b_n$,
where, for $i=1,\dots,n$, $b_i = 1$ if and only if
$v_i \in S$. Thus, each
independent subset of $\cy{n}{h}$ is associated with a circular binary string of length
$n$ such that the distance between any two $1$'s of the string is greater than $h$.
A \emph{Lucas cube of order $n$}, denoted $\Lambda_n$, is defined as the graph whose vertices
are the binary strings of length $n$ without either two consecutive $1$'s or a $1$ in the first
and in the last position, and in which the vertices are adjacent when their Hamming distance
is exactly $1$ (see \cite{lucas}). For $h=1$ the Hasse diagram of the set of all circular binary strings
associated with independent subsets of $\cy{n}{h}$ ordered bitwise is $\Lambda_n$.
A generalization of the notion of Lucas cubes has been proposed
in \cite{gen_lucas}.
\begin{remark}Consider the \emph{generalized Lucas cubes} described in \cite{gen_lucas}, that is, the graphs $B_n(\widehat{\alpha})$ obtained from the $n$-cube $B_n$ of all binary strings of length $n$ by removing all vertices that have a \emph{circular containing} $\alpha$ as a substring (\ie, such that $\alpha$ is contained in the circular binary strings obtained by connecting first and last bits of the string). In this notation
the Lucas cube is $B_n(\widehat{11})$. It is not difficult to see that $\hc{n}{h}$ cannot be expressed, in general, in terms of
$B_n(\widehat{\alpha})$. Instead we have:

\smallskip\centerline{
$\hc{n}{2}=B_n(\widehat{11})\cap B_n(\widehat{101})\,,\
\hc{n}{3}=B_n(\widehat{11})\cap B_n(\widehat{101})\cap B_n(\widehat{1001})\,,\ \dots$}
\end{remark}

As far as we now, our $\hp{n}{h}$, and $\hc{n}{h}$ are new generalizations of Fibonacci and Lucas cubes, respectively.

\medskip
In the first part of the paper we evaluate $p_n^{(h)}$, \ie, the number of \emph{independent subsets}
of $\pa{n}{h}$, and $H_n^{(h)}$, \ie, the number of edges of $\hp{n}{h}$.
Our main result (Theorem \ref{th:main}) is that,
for $n,h\geq 0$, the sequence $H_n^{(h)}$ is obtained
by convolving the sequence $\underbrace{1,\dots,1}_{h},p_{0}^{(h)},p_{1}^{(h)},p_{2}^{(h)},\dots$
with itself.

In the second part of the paper we derive similar results for $q_n^{(h)}$, \ie, the number of \emph{independent subsets}
of $\cy{n}{h}$, and $M_n^{(h)}$, \ie, the number of edges of $\hc{n}{h}$.

\section{The independent subsets of powers of paths}\label{sec:ind sub of paths}

For $n, h, k\geq 0$, we denote by $p_{n,k}^{(h)}$ the number of  independent
$k$-subsets of $\pa{n}{h}$.
\begin{remark}
For $h=1$, $p_{n,k}^{(h)}$ counts the number of binary strings $\alpha\in \Gamma_n$
such that $H(\alpha,00\cdots0)=k$.
\end{remark}
\begin{lemma}
\label{lem:p_nk formula}
For $n, h, k \geq 0$,
\[
p_{n,k}^{(h)} = \binom{n-hk+h}{k}\ .
\]
\end{lemma}
This result is Theorem 1 of \cite{hoggatt}. Below we write down a different proof.
\begin{proof}
By Definition \ref{def:indipendent subset},
any two elements $v_i,v_j$ of an independent subset of $\pa{n}{h}$
must satisfy $|j-i|>h$. It is straightforward to check that
whenever $n- hk-h<0$, $p_{n,k}^{(h)} = 0 = \binom{n-hk+h}{k}$.
It is also immediate to see that when $n=h=0$ our lemma holds true.

Suppose now $n- hk-h\geq 0$. We can complete the proof of our lemma by establishing a bijection
between independent $k$-subset of $\pa{n}{h}$ and $k$-subsets
of a set with $(n-hk+h)$ elements. Let $\mathscr{K}$ be the set of all $k$-subsets of a set $B = \{b_{1},b_{2},\dots,b_{n-hk+h}\}$, and
$\mathcal{I}_k$ the set of all independent $k$-subsets of $\pa{n}{h}$.
Consider the map $f: \mathscr{K} \to \mathcal{I}_k$ such that, for any $S=\{b_{i_1},b_{i_2},\dots, b_{i_k}\} \in \mathscr{K}$, with $1\leq i_1<i_2<\cdots <i_k \leq n-hk+h$,
\[
f(\{b_{i_1},b_{i_2},\dots, b_{i_j}, \dots, b_{i_k}\}) = \{v_{i_1},v_{i_2 + h},\dots, v_{i_j + (j-1)h}, \dots, v_{i_k + (k-1)h}\}\,.
\]

\begin{claim}\label{claim:1} The map $f$ associates with each $k$-subset
$S=\{b_{i_1}, b_{i_2}, \dots, b_{i_k}\} \in \mathscr{K}$
an independent $k$-subset of $\pa{n}{h}$.
\end{claim}
To see this we first remark that $f(S)$ is a $k$-subset of $V(\pa{n}{h})$.
Furthermore, for each pair $b_{i_j},b_{i_{j+t}}\in S$, with $t>0$, we have
$$i_{j+t} + (j+t-1)h - (i_j + (j-1)h) = i_{j+t} - i_j + th > h\,.$$
Hence, by Definition \ref{def:h-path}, $(f(b_{i_j}),f(b_{i_{j+t}}))=(v_{i_j + (j-1)h},v_{i_{j+t} + (j+t-1)h})\notin E(\pa{n}{h})$. Thus,
$f(S)$ is an independent subset of $\pa{n}{h}$.

\begin{claim}\label{claim:2} The map $f$ is bijective.
\end{claim}
It is easy to see that $f$ is injective.
Then, we consider the map $f^{-1}: \mathcal{I}_k \to \mathscr{K}$ such that, for any $S=\{v_{i_1},v_{i_2},\dots, v_{i_k}\} \in \mathcal{I}$, with $1\leq i_1<i_2<\cdots <i_k \leq n$,
\[
f^{-1}(\{v_{i_1},v_{i_2},\dots, v_{i_j}, \dots, v_{i_k}\}) = \{b_{i_1},b_{i_2 - h},\dots, b_{i_j - (j-1)h}, \dots, b_{i_k - (k-1)h}\}\,.
\]
Following the same steps as for $f$, one checks that $f^{-1}$
is injective. Thus, $f$ is surjective.

\medskip
By Claims \ref{claim:1} and \ref{claim:2} we have established a bijection
between independent $k$-subsets of $\pa{n}{h}$ and $k$-subsets of a set with
$(n-hk+h) \geq 0$ elements. The lemma is proved.
\end{proof}

\noindent The coefficients $p_{n,k}^{(h)}$ also enjoy the following property:
$p_{n,k}^{(h)} = p_{n-k+1,k}^{(h-1)}$.

\smallskip
For $n,h\geq0$, the number of  all independent subsets of $\pa{n}{h}$ is
\[
p_n^{(h)}=\sum_{k\geq0}p_{n,k}^{(h)}=\sum_{k=0}^{\lceil n/(h+1)\rceil}p_{n,k}^{(h)}=\sum_{k=0}^{\lceil n/(h+1)\rceil}\binom{n-hk+h}{k}\,.
\]

\begin{remark} Denote by $F_n$ the $n^{th}$ element of the Fibonacci sequence $F_1=1$, $F_2=1$, and $F_i=F_{i-1}+F_{i-2}$, for $i>2$. Then, $p_n^{(1)}=F_{n+2}$ is the number of elements of the Fibonacci cube of order $n$.
\end{remark}

The following, simple fact is crucial for our work.

\begin{lemma}
\label{lem:p_n recurrence}
For $n, h \geq 0$,
\[
p_n^{(h)} = \begin{cases}
n+1 & \text{if }\ n \leq h+1\,, \\
p_{n-1}^{(h)}+p_{n-h-1}^{(h)} & \text{if }\ n>h+1\,.
\end{cases}
\]
\end{lemma}
A proof of this Lemma can be also obtained using the first part
of \cite[Proof of Theorem 1]{hoggatt}.
\begin{proof}
For $n \leq h+1$, by Definition \ref{def:indipendent subset},
the independent subsets of $\pa{n}{h}$ have no more than $1$ element.
Thus, there are $n+1$ independent subsets of $\pa{n}{h}$.

\smallskip
Consider the case $n > h+1$.
Let $\mathcal{I}$ be
the set of all independent subsets of $\pa{n}{h}$, let
$\mathcal{I}_{in}$ be the set of the independent subsets of $\pa{n}{h}$
that contain $v_n$, and let $\mathcal{I}_{out}=\mathcal{I}\setminus\mathcal{I}_{in}$.
The elements
of $\mathcal{I}_{out}$ are in one-to-one correspondence with the
$p_{n-1}^{(h)}$ independent subsets of $\pa{n-1}{h}$, and
those of $\mathcal{I}_{in}$ are in one-to-one correspondence with the
$p_{n-h-1}^{(h)}$ independent subsets of $\pa{n-h-1}{h}$.
\end{proof}

\section{The poset of independent subsets of powers of paths}
\label{sec:hasse of paths}
Figure \ref{fig:P_3-4_h} shows a few Hasse diagrams $\hp{n}{h}$. Notice that,
as stated in the introduction, for each $n$, $\hp{n}{1}$ is the Fibonacci cube $\Gamma_n$.

\begin{figure}[h!]
  \centerline{\includegraphics[scale=0.6]{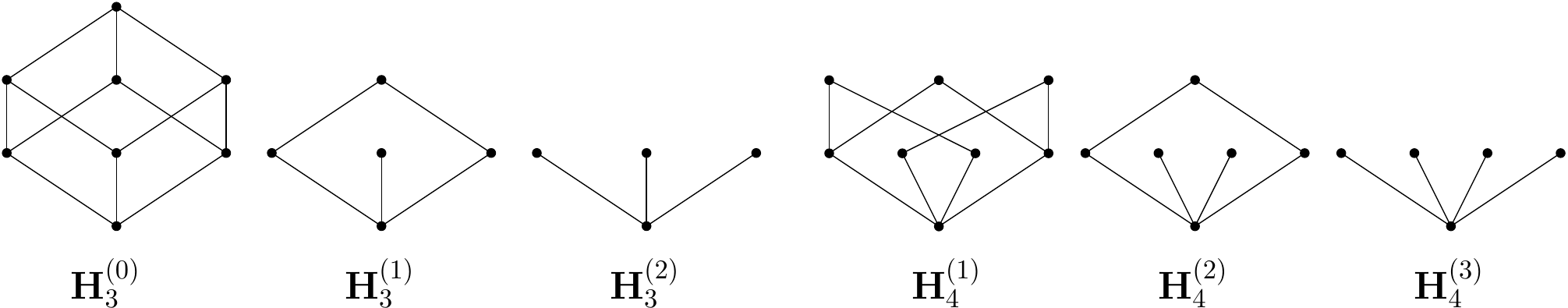}}
  \caption[Some $\hp{n}{h}$]
    {Some $\hp{n}{h}$.}
  \label{fig:P_3-4_h}
\end{figure}

Let $H_n^{(h)}$ be the number of edges of $\hp{n}{h}$.
Noting that in $\hp{n}{h}$ each non-empty independent $k$-subset covers
exactly $k$ independent $(k-1)$-subsets, we can write

\begin{equation}\label{eq:H_nh formula}
H_n^{(h)}= \sum_{k=1}^{\lceil n/(h+1)\rceil}kp_{n,k}^{(h)}= \sum_{k=1}^{\lceil n/(h+1)\rceil}k{{n-hk+h}\choose k}\ .
\end{equation}

\begin{remark}
For $h=1$, $H_n^{(h)}$ counts the number of edges of $\Gamma_n$.
\end{remark}

Let now $T_{k,i}^{(n,h)}$ be the number of independent $k$-subsets of $\pa{n}{h}$ containing the vertex $v_i$,
and let, for $h,k\geq 0$, $n \in \mathbb{Z}$,
$\bar{p}_{n,k}^{(h)} = \begin{cases}
p_{0,k}^{(h)} & \text{if }\ n < 0\,, \\
p_{n,k}^{(h)} & \text{if }\ n \geq 0\,.
\end{cases}$

\begin{lemma}
\label{lem:T formula}
For $n,h,k \geq 0$, and $1\leq i \leq n$,
\[
T_{k,i}^{(n,h)}=\sum_{r=0}^{k-1}\, \bar{p}_{i-h-1,r}^{(h)} \ \bar{p}_{n-i-h,k-1-r}^{(h)}\,.
\]
\end{lemma}
\begin{proof}
No independent subset of $\pa{n}{h}$ containing $v_i$ contains any of
the elements $v_{i-h}, \dots,v_{i-1},v_{i+1},\dots,v_{i+h}$.
Let $r$ and $s$ be non-negative integers whose sum is $k-1$. Each independent $k$-subset of $\pa{n}{h}$ containing $v_i$
can be obtained by adding $v_i$ to a $(k-1)$-subset $R\cup S$ such that

\noindent(a) $R\subseteq \{v_1,\dots,v_{i-h-1}\}$ is an independent $r$-subset of $\pa{n}{h}$;

\noindent(b) $S\subseteq \{v_{i+h+1},\dots,v_n\}$ is an independent $s$-subset of $\pa{n}{h}$.

\smallskip
Viceversa, one can obtain each of this pairs of subsets by removing $v_i$ from an independent
$k$-subset of $\pa{n}{h}$ containing $v_i$.
Thus, $T_{k,i}^{(n,h)}$ is obtained by counting independently
the subsets of type (a) and (b). Noting that the subsets of type (b) are
in bijection with the independent $s$-subsets of
$\pa{n-i-h}{h}$,
the lemma is proved.
\end{proof}

\begin{remark}
$T_{k,i}^{(n,1)}$ counts the number of strings $\alpha=b_1 b_2 \cdots b_n\in\Gamma_n$ such that:
(i) $H(\alpha,00\cdots0)=k$, and (ii) $b_i=1$.
\end{remark}

In order to obtain our main result, we prepare a lemma.

\begin{lemma}
\label{lem:somma tabelle T}
For positive $n$,
\[
\sum_{k=1}^{\lceil n/(h+1)\rceil}\sum_{i=1}^n T_{k,i}^{(n,h)} = H_n^{(h)}\,.
\]
\end{lemma}
\begin{proof}
The inner sum
counts the number of $k$-subsets exactly $k$ times, one for each element of the subset.
That is,
$\sum_{i=1}^n T_{k,i}^{(n,h)} = k p_{n,k}^{(h)}$.
The lemma follows directly from Equation (\ref{eq:H_nh formula}).
\end{proof}

Next we introduce a family
of Fibonacci-like sequences.

\begin{definition}\label{def:gen fibonacci}
For $h\geq 0$, and $n\geq 1$, we define the \emph{$h$-Fibonacci sequence}
$\mathcal{F}^{(h)}=\{F_n^{(h)}\}_{n\geq 1}$ whose elements are
\[
F_n^{(h)} = \begin{cases}
1 & \text{if }\ n \leq h+1\,, \\
F_{n-1}^{(h)}+F_{n-h-1}^{(h)} & \text{if }\ n>h+1.
\end{cases}
\]
\end{definition}

From Lemma \ref{lem:p_n recurrence}, and setting
for $h\geq 0$, and $n \in \mathbb{Z}$,
$\bar{p}_{n}^{(h)} = \begin{cases}
p_{0}^{(h)} & \text{if }\ n < 0\,, \\
p_{n}^{(h)} & \text{if }\ n \geq 0\,,
\end{cases}$
we have that,
\begin{equation}\label{eq:F and p}
F_i^{(h)}=\bar{p}_{i-h-1}^{(h)}\,, \ \ \text{for each}\ \  i\geq 1\,.
\end{equation}

\smallskip
Thus, our Fibonacci-like sequences are obtained by adding a prefix of $h$
ones to the sequence $p_{0}^{(h)},p_{1}^{(h)},\dots\ $. Therefore, we have:
\begin{itemize}
\item $\mathcal{F}^{(0)}=1,2,4,\dots,2^{n},\dots$;
\item $\mathcal{F}^{(1)}$ is the Fibonacci sequence;
\item more generally, $\mathcal{F}^{(h)}=\underbrace{1,\dots,1}_{h},p_{0}^{(h)},p_{1}^{(h)},p_{2}^{(h)},\dots$.
\end{itemize}

In the following, we use the discrete convolution operation $\ast$, as follows.
\begin{equation}\label{eq:convolution}
\left(\mathcal{F}^{(h)}\ast \mathcal{F}^{(h)}\right)(n)\doteq \sum_{i=1}^n F_{i}^{(h)} F_{n-i+1}^{(h)}
\end{equation}

\begin{theorem}
\label{th:main}
For $n,h\geq 0$, the following holds.
\[
H_n^{(h)} = \left(\mathcal{F}^{(h)}\ast \mathcal{F}^{(h)}\right)(n)\,.
\]
\end{theorem}
\begin{proof}
The sum $\sum_{k=1}^{\lceil n/(h+1)\rceil} T_{k,i}^{(n,h)}$
counts the number of independent subsets of $\pa{n}{k}$ containing $v_i$.
We can also obtain such a value by counting
the independent subsets of both $\{v_1, \dots, v_{i-h-1}\}$, and $\{v_{i+h+1}, \dots, v_{n}\}$.
Thus,  we have:
\[
\sum_{k=1}^{\lceil n/(h+1)\rceil} T_{k,i}^{(n,h)} = \bar{p}_{i-h-1}^{(h)}\, \bar{p}_{n-h-i}^{(h)}\,.
\]
Using Lemma \ref{lem:somma tabelle T} we can write
\[
H_n^{(h)} = \sum_{k=1}^{\lceil n/(h+1)\rceil}\sum_{i=1}^n T_{k,i}^{(n,h)} =
\sum_{i=1}^n\sum_{k=1}^{\lceil n/(h+1)\rceil} T_{k,i}^{(n,h)} =
\sum_{i=1}^n \bar{p}_{i-h-1}^{(h)}\, \bar{p}_{n-h-i}^{(h)}.
\]
By Equation (\ref{eq:F and p}) we have
$\sum_{i=1}^n \bar{p}_{i-h-1}^{(h)}\, \bar{p}_{n-h-i}^{(h)}=\sum_{i=1}^n F_{i}^{(h)} F_{n-i+1}^{(h)}\,.$
By (\ref{eq:convolution}), the theorem is proved.
\end{proof}

\begin{remark} For $h=1$, we obtain the number of edges of $\Gamma_n$ by using Fibonacci numbers:
\[
H_n^{(h)} = \sum_{i=1}^n F_{i} F_{n-i+1}\,.
\]
The latter result is \cite[Proposition 3]{mediannature}.
\end{remark}

\section{The independent subsets of powers of cycles}\label{sec:ind sub of cycles}

For $n, h, k\geq 0$, we denote by $q_{n,k}^{(h)}$ the number of  independent
$k$-subsets of $\cy{n}{h}$.
\begin{remark}
For $h=1$, $n>1$, $q_{n,k}^{(h)}$ counts the number of binary strings $\alpha\in \Lambda_n$
such that $H(\alpha,00\cdots0)=k$.
\end{remark}
\begin{lemma}\label{lem:q_nk formula}
For $n, h\geq 0$, and $k>1$,
\[
q_{n,k}^{(h)} = \frac{n}{k}\binom{n-hk-1}{k-1}\ .
\]
Moreover, $q_{n,0}^{(h)}=1$, and $q_{n,1}^{(h)}=n$, for each $n, h\geq 0$.
\end{lemma}
\begin{proof}
Fix an element $v_i \in V(\cy{n}{h})$, and let $n>2h$.
Any independent subset of $\cy{n}{h}$
containing $v_i$ does not contain the $h$ elements preceding  $v_i$ and the $h$
elements following $v_i$. Thus, the number of
independent $k$-subsets of $\cy{n}{h}$ containing $v_i$ equals
\[
p_{n-2h-1,k-1}^{(h)}=\binom{n-hk-1}{k-1}\ .
\]
The total number of independent $k$-subsets of $\cy{n}{h}$ is obtained
by multiplying $p_{n-2h-1,k-1}^{(h)}$ by $n$, then dividing it by $k$ (each
subset is counted $k$ times by the previous proceeding).
The case $n\leq 2h$, as well as the cases $k=0,1$, can be easily verified.
\end{proof}

For $n,h\geq0$, the number of  all independent subsets of $\cy{n}{h}$ is
\begin{equation}\label{eq:q_n def}
q_n^{(h)}=\sum_{k\geq0}q_{n,k}^{(h)}=\sum_{k=0}^{\lceil n/(h+1)\rceil}q_{n,k}^{(h)}\ ,
\end{equation}
\begin{remark} Denote by $L_n$ the $n^{th}$ element of the Lucas sequence $L_1=1$, $L_2=3$, and $L_i=L_{i-1}+L_{i-2}$, for $i>2$. Then, for $n>1$, $q_n^{(1)}=L_{n}$ is the number of elements of the Lucas cube of order $n$.
\end{remark}

The coefficients $q_n^{(h)}$ satisfy a recursion that closely resemble that of Lemma 2.2.
\begin{lemma}
\label{lem:q_n recurrence}
For $n, h \geq 0$,
\begin{equation}\label{eq:q_n recurrence}
q_n^{(h)} = \begin{cases}
n+1 & \text{if }\ n\leq 2h+1\,, \\
q_{n-1}^{(h)}+q_{n-h-1}^{(h)} & \text{if }\ n>2h+1.
\end{cases}
\end{equation}
\end{lemma}
\begin{proof}
The case $n\leq 2h+1$ can be easily checked. Let $n> 3h+2$, and let $\mathcal{I}$ be the set of the independent subsets of $\cy{n}{h}$. Let $\mathcal{I}_{in}$ be
the subset of these subsets that (i) do not contain $v_n$, and that  (ii) contain
no one of the following pairs: $(v_1, v_{n-h}), (v_2, v_{n-h+1}), \dots, (v_h, v_{n-1})$. Furthermore
let $\mathcal{I}_{out}$ be the subset of the remaining independent subsets of $\cy{n}{h}$.

\smallskip
It is easy to see that the elements of $\mathcal{I}_{in}$ are exactly the independent subsets
of $\cy{n-1}{h}$. Indeed, $v_n$ is not a vertex of $\cy{n-1}{h}$ and the vertices
of pairs $(v_1, v_{n-h})$, $(v_2, v_{n-h+1})$, $\dots$, $(v_h, v_{n-1})$ are connected in $\cy{n-1}{h}$.
On the other hand, to show that
\[
|\mathcal{I}_{out}|=q_{n-h-1}^{(h)}
\]
we argue as follows. First we recall (see the proof of Lemma \ref{lem:q_nk formula})
that the number of independent $k$-subsets of
$\cy{n}{h}$ that contain $v_n$ is $p_{n-2h-1,k-1}^{(h)}$.
Secondly we obtain that the number of independent $k$-subsets
of $\cy{n}{h}$ containing one of the pairs $(v_1, v_{n-h})$, $(v_2, v_{n-h+1})$, $\dots$,
$(v_h, v_{n-1})$ is $h p_{n-3h-2,k-2}^{(h)}$. To see this, consider the pair $(v_1,v_{n-h})$. The independent
subsets containing such a pair do not contain the $h$ vertices from $v_n-h+1$ to $v_n$, do
not contain the $h$ vertices from $v_2$ to $v_{h+1}$, and do not contain the $h$ vertices from
$v_{n-2h}$ to $v_{n-h-1}$. Thus, the removal of such vertices and of the vertices $v_1$ and $v_{n-h}$
turns $\cy{n}{h}$ into $\pa{n-3h-2}{h}$. Hence we can obtain all the independent $k$-subsets of
$\cy{n}{h}$ that contain the pair $(v_1,v_{n-h})$ by simply adding these two vertices to one of the
$p_{n-3h-2,k-2}^{(h)}$ independent $k-2$-subsets of $\pa{n-3h-2}{h}$. Same reasoning
can be carried out for any other one of the pairs: $(v_2, v_{n-h+1})$, $\dots$,
$(v_h, v_{n-1})$.

\smallskip
Using Lemmas \ref{lem:p_nk formula} and \ref{lem:q_nk formula} one can easily derive
that
\[
p_{n-2h-1,k-1}^{(h)}+hp_{n-3h-2,k-2}^{(h)}=q_{n-h-1,k-1}^{(h)}\,.
\]
Hence, we derive the size of $\mathcal{I}_{out}$:
\[
|\mathcal{I}_{out}|=q_{n-h-1}^{(h)}=\sum_{k\geq 1}p_{n-2h-1,k-1}^{(h)}\ + h\sum_{k\geq2}p_{n-3h-2,k-2}^{(h)}\,.
\]
Summing up we have shown that $|\mathcal{I}|=|\mathcal{I}_{in}|+|\mathcal{I}_{out}|$, that is
\[
q_n^{(h)}=q_{n-1}^{(h)}+q_{n-h-1}^{(h)}\,.
\]

The proof of the case $2h+1< n \leq 3h+2$
is obtained in a similar way, observing that $|\mathcal{I}_{out}|=n-h$, and that
$n-h-1\leq 2h+1$.
\end{proof}

\section{The poset of independent subsets of powers of cycles}
\label{sec:hasse of cycles}

Figure \ref{fig:Q_3-4_h} shows a few Hasse diagrams $\hc{n}{h}$. Notice that,
as stated in the introduction, for each $n$, $\hc{n}{1}$ is the Lucas cube $\Lambda_n$.

\begin{figure}[h!]
  \centerline{\includegraphics[scale=0.6]{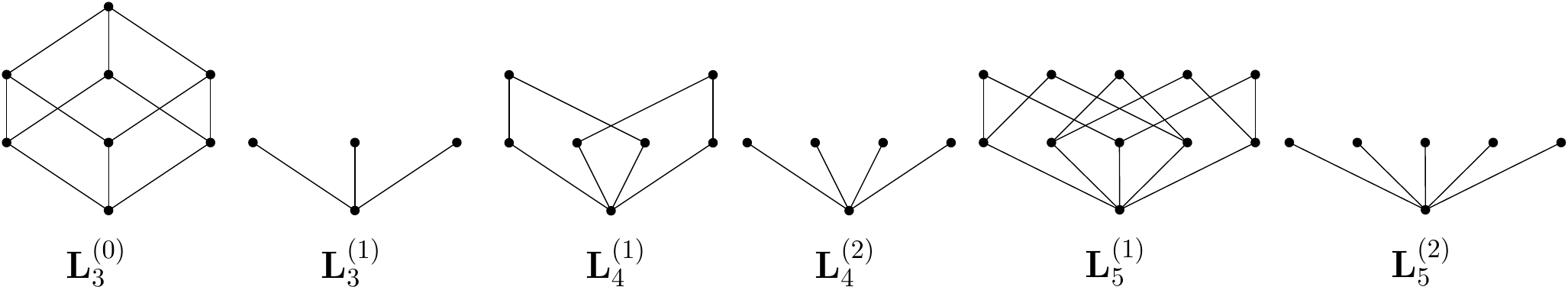}}
  \caption[Some $\hc{n}{h}$]
    {Some $\hc{n}{h}$.}
  \label{fig:Q_3-4_h}
\end{figure}

Let $M_n^{(h)}$ be the number of edges of $\hc{n}{h}$.
As done in Section \ref{sec:hasse of paths} for the case of paths, we immediately provide
a formula for $M_n^{(h)}$:
\begin{equation}\label{eq:M_nh formula}
M_n^{(h)}= \sum_{k=0}^{\lceil n/h+1 \rceil}kq_{n,k}^{(h)} = n \sum_{k=0}^{\lceil n/h+1 \rceil}{\binom{n-hk-1}{k-1}}\ .
\end{equation}
\begin{remark}
For $h=1$, $n>1$, $M_n^{(h)}$ counts the number of edges of $\Lambda_n$.
As shown in \cite[Proposition 4(ii)]{lucas}, $M_n^{(h)}=n F_{n-1}$.
\end{remark}
\smallskip
As shown in the proof of Lemma \ref{lem:q_nk formula}, the value
\[
p_{n-2h-1,k-1}^{(h)}=\binom{n-hk-1}{k-1}
\]
is the analogue of the coefficient $T_{k,i}^{(n,h)}$:  in the case of cycles we have no dependencies on $i$,
because each choice of vertex is equivalent.
We can obtain $M_n^{(h)}$ in terms of a fibonacci-like sequence, as follows.
\begin{proposition}
For $n>h\geq 0$, the following holds.
\[
M_n^{(h)} = n F_{n-h}^{(h)}\,.
\]
\end{proposition}
\begin{proof}
Using Equation (\ref{eq:F and p}) we obtain:
\[
M_n^{(h)} = n\sum_{k=1}^{\lceil n/(h+1)\rceil} \bar{p}_{n-2h-1,k-1}^{(h)} =
n \bar{p}_{n-2h-1}^{(h)}= n F_{n-h}^{(h)}\,.
\]
\end{proof}

\providecommand{\bysame}{\leavevmode\hbox to3em{\hrulefill}\thinspace}
\providecommand{\MR}{\relax\ifhmode\unskip\space\fi MR }
\providecommand{\MRhref}[2]{%
  \href{http://www.ams.org/mathscinet-getitem?mr=#1}{#2}
}
\providecommand{\href}[2]{#2}

\end{document}